\documentclass[11pt]{article}
\usepackage[utf8]{inputenc}

\usepackage{epsfig, amssymb, amsmath, geometry}
\usepackage{latexsym, graphics, graphicx}
\usepackage{proof}
\usepackage{algorithmic}
\usepackage{algorithm}
\usepackage{enumitem}
\usepackage{amsthm} 
\usepackage{amsxtra}
\usepackage{epsfig}
\usepackage{float}
\usepackage{authblk}
\usepackage{url}
\usepackage{mathptmx} 
\usepackage{enumerate}
\usepackage[numbers]{natbib}
\usepackage{fancyvrb}
\usepackage{marginnote}
\usepackage{mathtools}
\usepackage{wasysym}
\usepackage{verbatim}%to add block comments \begin{comment}%

\usepackage{todonotes}
\usepackage{enumitem}
\newlist{todolist}{itemize}{2}
\setlist[todolist]{label=$\square$}
\newcommand{\ignore}[1]{}

\newcommand{\flr}{\rightarrow}

\newcommand{\dom}{\mathcal{D}om}

\newcommand{\vran}{\mathcal{V\!\!R}\!\mathit{an}}

\newtheorem{lemma}[section]{Lemma}
\newtheorem{theorem}{Theorem}

\makeatletter
\newcommand*{\rom}[1]{\expandafter\@slowromancap\romannumeral #1@}
\makeatother

\title{On Problems Dual to Unification}
\author[1]{Z\"{u}mr\"{u}t Ak{\c c}am}
\author[1]{Daniel S. Hono II}
\author[1]{Paliath Narendran}
\affil[1]{University at Albany, SUNY\\ Albany, NY, US}
\affil[ ]{\texttt{\{zakcam, dhono, pnarendran\}@albany.edu}}

\date{}

\begin{document}
\maketitle

\begin{abstract}
  In this paper, we investigate a problem dual to the unification
  problem, namely the \emph{Common Term~(CT)} problem for string
  rewriting systems. Our main motivation is computing
  \emph{fixed points} in systems, such as \emph{loop invariants}
  in programming languages. We show that the fixed point problem
  is reducible to the common term problem. We also prove that the
  common term problem is undecidable for the class of \emph{dwindling} string
  rewriting systems.
\end{abstract}

\section{Introduction}
Unification, with or without background theories such as
associativity and commutativity, is an area of great theoretical
and practical interest. The latter problem, called \emph{equational}
or \emph{semantic} unification, has been studied from several different angles.
Here we investigate some problems that can be viewed as \emph{dual}
to the unification problem.
Our main motivation for this work is theoretical, but, as explained below,
we are also interested in a practical application that is shared by
many fields.

In every major research field, there are variables or other parameters that changes over
time. These variables are modified --- increased
or decreased --- as a
result of a change in the environment. 
Computing \emph{invariants,} or expressions whose values do not change
under a transformation, is very important in many areas 
such as Physics, e.g., invariance under the \emph{Lorentz} transformation.

In Computer Science, the issue of obtaining invariants arises
in \emph{axiomatic semantics} or \emph{Floyd-Hoare semantics},
in the context of formally proving a loop to be correct. A
\emph{loop invariant} is a
condition, over the program variables, 
that holds before and after each iteration. Our research
is partly motivated by the related question of
finding expressions, called \emph{fixed points}, whose values
will be the same before and after each iteration, i.e., will
remain unchanged as long as the iteration goes on. For instance,
for a loop whose body is \[ \text{\tt X = X + 2; Y = Y - 1;} \] the
value of the expression {\tt X + 2Y} is a fixed point.

We can formulate this problem in terms of properties of substitutions
\emph{modulo} a term rewriting system. One straightforward formulation
is as follows:\\

\noindent
{\large\bf Fixed Point Problem (FP)}

\begin{description}[align=left]
\item[Input:] A substitution $\theta$ and an equational theory~$E$.
\item[Question:] Does there exist a non-ground term~$t \, \in \, T(Sig(E), \, {\dom}(\theta))$ 
such that~$\theta (t) \, \approx_E^{} t$?
\end{description}

\vspace{0.05in}
\noindent
\underline{Example 1}: Suppose~$E$ is a theory of integers which contains linear
arithmetic. Let $\theta = \{x \mapsto x-2, \, y \mapsto y+1\}$ and we would like to find a term~$t$ such that 
$\theta (t) \approx_E^{} t$. Note that $x + 2y$ is such a term, since \[ \theta (x+2y) =
(x - 2) + 2*(y + 1) \approx_E^{} x+2y \] \\

We plan to explore two related formulations, both of which can be 
viewed as \emph{dual} to the well-known unification problem. Unification
deals with solving symbolic equations: thus a typical input would be either two
terms, say $s$ and $t$, or an equation~$s \approx_{}^? t$. The task is to
find a substitution such that $\theta (s) \approx \theta(t)$.
For example, given two terms $s_1^{} = f(a, y)$ and $s_2^{} = f(x, b)$, where $f$~is
a binary function symbol, $a$ and $b$ are constants, and $x$ and $y$
are variables, the substitution $\sigma = \{ x \mapsto a, ~ y \mapsto
b\}$ unifies $s_1^{}$ and~$s_2^{}$, or equivalently, $\sigma$ is a unifier
for the equation~$s_1^{} =_{}^? s_2^{}$.

%\pagebreak

There are two ways to ``dualize'' the unification problem:\\

\noindent
{\large\bf Common Term Problem (CT)}:

\begin{description}[align=left]
\item [Input:] Two ground substitutions $\theta_1^{}$ and $\theta_2^{}$, and an
  equational theory~$E$. (i.e., $\vran(\theta_1) = \emptyset$ and $\vran(\theta_2) = \emptyset$ )
\item [Question:] Does there exist a non-ground term~$t \, \in \, 
T(Sig(E), \, \dom(\theta_1) \cup \dom(\theta_2))$ such that
\mbox{$\theta_1^{} (t) \, \approx_E^{} \theta_2^{} (t)$}?
\end{description} %$t \, \in \, T(Sig(E), \, \dom(\theta))$%

\ignore{
We will be considering equational theories that 
are decomposable into a set of identities~$Ax$ and
a set of rewrite rules such that \emph{equational} rewriting 
modulo~$Ax$ is convergent. The most widely used case of
equational rewriting is where $Ax$ consists of associativity and
commutativity axioms~(\emph{AC}).
The key concepts are defined in the next section.
}

\vspace{0.05in}
\noindent
\underline{Example 2}: Consider the two substitutions $\theta_{1} = \{ x
\mapsto p(a), \, y \mapsto p(b) \}$ and $\theta_{2} = \{ x \mapsto a,
\, y \mapsto b \}$. If we take the term rewriting system~$R_1^{lin}$ in the appendix
as our background equational theory~$E$, then there exists a common term $t =
x-y$ that satisfies $\theta_{1}(t) \approx_E^{}
\theta_{2}(t)$. \[ \theta_1 ( x-y ) \approx_E^{} p(a)- p(b) \approx_E^{}
a-b \] and \[\theta_2 ( x-y ) \approx_E^{} a-b \]
\vspace{0.05in}

\pagebreak
We can easily show that the fixed point problem can be reduced to the CT problem. 
\begin{lemma}
The fixed point problem is reducible to the common term problem.
\end{lemma}
\begin{proof}
Let $\theta_2$ be
the empty substitution. 
Assume that the fixed point problem has a solution, i.e., there exists
a term $t$ such that $\theta (t) \, \approx_E^{} t$. Then the CT problem
for $\theta$ and $\theta_2$ has a solution since
$\theta_2(t) \approx_E^{} t$ (because $\theta_2 (s) = s$ for all~$s$). The ``only if'' part is trivial, again
because $\theta_2 (s) = s$ for all~$s$.

Alternatively, suppose
that $\dom(\theta)$ consists of $n$~variables, where $n \geq 1$. If we
map all the variables in $\vran(\theta)$ to new constants, this will
create a ground substitution $\theta_1 = \{x_1
\mapsto a_1,\; x_2 \mapsto a_2,\; ..., \;x_n \mapsto
a_n\}$. $\theta_1$ will be the one of the substitutions for the CT
problem. The other substitution, $\theta_2$, is the composition of the
substitutions $\theta$ and $\theta_1$. The substitution $\theta_1$
will replace all of the variables in~$\vran(\theta)$ with the new
constants, thus making $\theta_2$ a ground substitution. Now
if $\theta (t) \, \approx_E^{} \, t$, then
$\theta_2 (t) = \theta_1 ( \theta (t) ) \approx_E^{} \theta_1( t )$; in other words,
$t$ is a solution to the common term problem.

The ``only if'' part can also be explained in terms of the composition
above. Suppose that $\theta_1(s)$ and $\theta_2(s)$ are equivalent,
i.e., $\theta_1(s) \, \approx_E^{} \, \theta_2(s)$ for some~$s$. Since
$\theta_2 = \theta_1 \circ \theta$, the equation can be rewritten as
$\theta_1(\theta(s)) \approx_E^{} \theta_1(s)$.
Since $a_1, \, \ldots , \, a_n$ are new constants and are not included in the signature of
the theory, for all $t_1$ and $t_2$, $\theta_1(t_1) \, \approx_E^{}
\theta_1(t_2)$ holds if and only if $t_1\approx_E^{} t_2$
(See~\cite{Term}, Section~4.1, page~60)
Thus $\theta_1(\theta(s)) \approx_E^{} \theta_1(s)$ implies that $\theta(s) \approx_E^{} s$, making
$s$~a fixed point. 
\end{proof}

\noindent
{\large\bf Common Equation Problem (CE)}:

\begin{description}[align=left]
\item[Input:] Two substitutions $\theta_1^{}$ and $\theta_2^{}$ with the \emph{same domain}, and an
  equational theory~$E$. 
\item[Question:] Does there exist a non-ground, non-trivial $(t_1
  \not\approx_E^{} t_2)$ equation~$t_1^{} \approx_{}^? t_2^{}$, where
  $t_1^{}, t_2^{} \, \in \, T(Sig(E), \, \dom(\theta_1^{}))$ such that
  both~$\theta_1^{}$ and~$\theta_2^{}$ are \emph{E}-unifiers
  of~$t_1^{} \approx_{}^? t_2^{}$?\\
\end{description} %$t_1^{}, t_2^{} \, \in \, T(Sig(E), \, \dom(\theta_1^{}))$%
By trivial equations, we mean equations which are identities
in the equational theory~$E$, i.e., an equation $s \approx_{E}^? t$ is
trivial if and only if $s \approx_{E}^{} t$.
We exclude this type of trivial equations in the formulation of
this question.

%We called the first problem in this dualization as the \emph{common
% term (CT) problem}, while the latter is the \emph{common equation (CE)
% problem}.

%We will investigate these two problems for 3~cases, the $\omega\Delta$
%case, the string rewriting case and linear arithmetic.

\vspace{0.05in}
\noindent
\underline{Example 3}: Let $\, E ~ = ~ \{ p(s(x)) \approx x, \; s(p(x)) \approx x \}$.
Given two substitutions $\theta_{1} = \{ x_1
\mapsto s(s(a)), \, x_2 \mapsto s(a) \}$ and $\theta_{2} = \{ x_1
\mapsto s(a), \, x_2 \mapsto a \}$, we can see that $\theta_{1}(t_1)
\approx_E^{} \theta_{1}(t_2)$ and $\theta_{2}(t_1) \approx_E^{}
\theta_{2}(t_2)$, with the equation \[ p(x_1) \approx_E^{} x_2\] However,
there is no term~$t$ on which the substitutions agree, i.e., there
aren't any solutions for the common term problem in this example.
Thus, CT and CE problems are not equivalent as we observe in the example above.\\[-7pt]

In this document we will discuss (and survey) these three problems for 
the string rewriting case.

\section{Definitions}
We start by presenting some notation and definitions on term rewriting
systems and particularly string rewriting systems. Only some
definitions are given in here, but for more details, refer to the
books~\cite{Term} for term rewriting systems and to~\cite{Botto} for
string rewriting systems.

A signature $\Sigma$ consists of finitely many ranked function symbols.
Let $X$~be a (possibly infinite) set of variables. The set
of all terms over~$\Sigma$ and~$X$ is denoted~as $T(\Sigma, X)$. 
The set of \emph{ground terms}, or terms with no variables
is denoted~$T(\Sigma)$.
A term
rewriting system~(TRS) is a set of rewrite rules that are defined on
the signature~$\Sigma$, in the form of $l \rightarrow r$, 
where $l$ and $r$
are called the left- and right-hand-side (\emph{lhs} and \emph{rhs}) of
the rule, respectively. The rewrite relation induced by
a term rewriting system~$R$ is denoted by~$\rightarrow_R^{}$. 
The reflexive
and transitive closure of~$\rightarrow_R^{}$ is denoted
$\rightarrow_R^{*}$. A TRS~$R$
is called \emph{terminating} iff there is no infinite chain of terms.
A TRS $R$~is \emph{confluent} iff, for all terms $t$,
$s_1$, $s_2$, if $s_1$ and $s_2$ can be derived from~$t$, i.e.,
$s_1 \leftarrow_R^{*} t \rightarrow_R^{*} s_2$, then there exists a 
term~$t'$ such that $s_1 \leftarrow_R^{*} t' \rightarrow_R^{*} s_2$.  A TRS
$R$ is \emph{convergent\/} iff it is both terminating and confluent.

\medskip{} A term is \emph{irreducible} iff no rule of TRS $R$ can be
applied to that term. The set of terms that are irreducible modulo~$R$
is defined by $IRR(R)$ and also called as terms in their \emph{normal
  form}s.  A term~$t'$ is said to be an \emph{R-normal form} of a term
$t$, iff it is irreducible and reachable from~$t$ in a finite number
of steps; this can be written as $t \rightarrow_{R}^{!} t'$.

String rewriting systems are a restricted class of term rewriting
systems where all functions are unary. These unary operators, that are
defined by the symbols of a string, applied in the order in which
these symbols appear in the string, i.e., if $g, h \in \Sigma$, the
string $gh$ will be seen as the term $h(g(x))$ \footnote{
  It may be more common to view $gh$ as $g(h(x))$ with
  function application done in the reverse order.
  }. The set of all strings
over the alphabet $\Sigma$ is denoted by $\Sigma^*$ and the empty
string is denoted by the symbol $\lambda$. Thus the term rewriting system
$\{ p(s(x)) \rightarrow x , \; s(p(x)) \rightarrow x \}$ is equivalent to
the string-rewriting system \[ \{ sp \rightarrow \lambda , \;
ps \rightarrow \lambda \} \]

If $R$ is a string
rewriting system (SRS) over alphabet $\Sigma$, then the single-step
reduction on $\Sigma^*$ can be written as:

For any $u,v \in \Sigma^*$, $u \rightarrow_R v$ iff there exists 
a rule~$l \rightarrow r
\in R$ such that 
$u = xly$ and $v = xry$ for some~$x, y \in \Sigma^*$; 
i.e., \[ {\rightarrow}_R^{} \; = \; \{( xly, \, xry) \; \mid \; (l \rightarrow r) \in R, 
  \, x,y \in {{\Sigma}^*} \} \] 

For any string rewriting system~$R$ over $\Sigma$, the set of all
irreducible strings, $IRR(R)$, is a regular
language: in fact, $IRR(R) = \Sigma_{}^* \smallsetminus \{\Sigma_{}^*
l_1 \Sigma_{}^* \,\cup ... \cup \, \Sigma_{}^* l_n \Sigma_{}^*\}$,
where $l_1,\dots, l_n$ are the left-hand sides of the rules in~$T$.

%We shall also be needing a special kind of normal form for strings, modulo  any given SRS $T$. With that purpose, we define, following S\'{e}nizergues/ ~\cite{Seniz-PartialC}, a {\em leftmost-smallest\/} reduction  as follows: let~$\succ$ be a given total ordering on the alphabet~$\Sigma$  and ~$\succ_{l}^{}$ be its length~+~lexicographic extension\footnote{Senizergues refers to this as the {\em short-lex\/} ordering}. A rewrite step $x l y \, \rightarrow \, x r y$ is {\em leftmost-smallest\/} if and only if (a)~no proper prefix of $xl$ is reducible, (b)~no proper suffix of~$l$ is reducible, and (c)~if $l \rightarrow r'$ is another rule in the rewrite system, then $r' \succ_{l}^{} r$.  A string $w'$ is said to be a {\em leftmost-smallest\/} ({\em ls-\/}) {\em normal form\/} of a string $w$ iff $w \rightarrow_{}^! w'$ using only  leftmost-smallest rewrite steps.Given a terminating system~$T$, it can be checked that any string~$w$  has a  \emph{unique} normal form produced by leftmost-smallest rewrite steps alone, since every rewrite step is unique; this unique normal form will be denoted as $\rho_T^{} (w)$.  For a language~$L \subseteq \Sigma_{}^*$, $\hat{\rho_T}(L)$ denotes the set of leftmost-smallest normal forms of strings in~$L$, i.e., $\hat{\rho_T}(L) ~ = ~ \{ \, \rho_T^{} (w) ~ | ~ w \in L \, \}$.

\medskip{}
Throughout the rest of the paper, 
$a, b, c, \dots, h$  will denote elements of the alphabet~$\Sigma$, 
and $l, r, u, v, w, x,
y,z$ will denote strings over~$\Sigma$. Concepts such as {\em normal form,\/}  {\em terminating,\/} 
{\em confluent,\/} and {\em convergent\/} have the same definitions in the string rewriting
systems as they have for the term rewriting systems.
An SRS~$T$ is called {\em canonical\/} if and only if it is convergent and
{\em inter-reduced,\/} i.e., no lhs is a substring of another lhs.

A  string rewriting system $T$ is said to be:
\begin{itemize}
\item[-] {\em monadic\/} iff the rhs of each rule in $T$ is either a single
 symbol or the empty string, e.g., $abc \rightarrow b$. \par 
\item[-] {\em dwindling\/} iff, for every rule $l \flr r$ in $T$, the rhs $r$ 
  is a {\em proper prefix\/} of its lhs $l$, e.g., $abc \rightarrow ab$. \par
\item[-] \emph{length-reducing} iff $| l | > | r |$ for all rules~$l \flr r$ in~$T$, e.g., $abc \rightarrow ba$.
\end{itemize}

\ignore{
There are also other classifications of string rewriting systems such as
\emph{special} and \emph{optimally reducing}.
Refer to~\cite{Botto} for more information on those
restricted classes.
}

\section{Fixed Point Problem}
  
%General string rewriting systems without restrictions for the word problem is undecidable. Therefore, we use restrictions on string rewriting systems to explore the decidability of these problems. Such as monadic Church-Rosser Thue Systems are decidable for FP problem. In the paper by Otto, fixed point problem is defined in terms of right factor(RF) problem. In a nutshell, if $uw \longleftrightarrow^{*}_{T} v$, $RF(v, u) = \{w \in IRR(T)| uw \longleftrightarrow^{*}_{T} v\}$, Otto proves that $RF(v,u)$ is a regular language and a non-deterministic finite automaton can be built to process of possibilities of right factors.

Note that for string rewriting systems the fixed point problem is equivalent 
to the following problem:

\begin{description}[align=left]
\item[Input:] A string-rewriting system~$R$ on an alphabet~$\Sigma$, and 
a string~$\alpha \in \Sigma_{}^+$.
\item[Question:] Does there exist a string~$W$ such that 
$\alpha W \; \stackrel{*}{{\longleftrightarrow}_R} \; W$?
\end{description}

This is a particular case of the Common Term Problem discussed
in the next section and is thus
decidable in polynomial time for
finite, monadic and convergent string rewriting systems.
\ignore{
In their paper, Lemma 3.7 shows the decidability of
CT, but a slight modification to that lemma, such as setting $h$ to be
$\epsilon$ will be sufficient to show FP is decidable. Therefore the
following equation in Lemma 3.7, $\exists w \in \Sigma^* : gw
\leftrightarrow^{*}_S hw$, will be equivalent to $\exists w \in
\Sigma^* : gw \leftrightarrow^{*}_S w$, which is equivalent to Fixed
Point problem definition.} %%
It is also a particular case of the \emph{conjugacy problem}.
Thus for finite, length-reducing and convergent systems it is
decidable in~\textbf{NP}~\cite{NOW}. The \textbf{NP}-hardness proof
in~\cite{NOW} also applies in our particular case: thus the problem
is~\textbf{NP}-complete for finite, length-reducing and convergent
systems.

%However, the non-uniform case, where the finite, length-reducing and
%convergent system~$R$ is \emph{fixed,} \textbf{FP}~is solvable in polynomial
%time.

%\pagebreak

\section{Common Term Problem}
  
Note that for string rewriting systems the common term problem is equivalent 
to the following problem:

\begin{description}[align=left]
\item[Input:] A string-rewriting system~$R$ on an alphabet~$\Sigma$, and 
two strings~$\alpha, \beta \in \Sigma_{}^*$.
\item[Question:] Does there exist a string~$W$ such that 
$\alpha W \; \stackrel{*}{{\longleftrightarrow}_R} \; \beta W$?
\end{description}

This is also known as \textit{Common Multiplier Problem} which
has been shown to
be decidable in polynomial time for monadic and convergent
string-rewriting systems (see, e.g.,~\cite{OND98}, Lemma~3.7). It is also known that the CT problem is undecidable for convergent string
rewriting systems; in fact, Otto et al.~\cite{OND98} proved that the
CT problem is undecidable even for \emph{convergent and
  length-reducing} string rewriting systems.
  
In this paper, we focus on the decidability of the CT problem for
\emph{convergent and dwindling} string rewriting systems. The
dwindling convergent systems are especially important because they are
widely used in the field of protocol analysis; in particular, digital
signatures, one-way hash functions and standard axiomatization of
encryption and decryption. This class is also known as subterm
convergent theories in the literature ~\cite{abadi2006deciding, 
Baudet2005, ciocaba2009, cortier2009}. Tools such as TAMARIN
prover~\cite{Meier2013} and YAPA~\cite{Baudet2009} use
subterm-convergent theories since these theories have nice properties
(e.g., finite basis property~\cite{chevalier2010compiling}) and
decidability results ~\cite{abadi2006deciding}.

\subsection{Dwindling CT problem}

We show that the CT (Common Term) problem is undecidable
for string rewriting systems that are dwindling and convergent.
We define CT as the following decision problem:

\begin{description}[align=left]
\item[Given:] A finite, non-empty alphabet $\Sigma$, strings $\alpha, \beta \in \Sigma^*$ and a dwindling, convergent string rewriting system $S$.
\item[Question:] Does there exist a string $W \in \Sigma^*$ such that $\alpha W \approx_S^{} \beta W$?
\end{description}

Note that interpreting concatenation the other way, i.e., $ab$ as $a(b(x))$, will make this
a \emph{unification} problem.

We show that Generalized Post Correspondence Problem ($GPCP$) reduces to the CT problem, where $GPCP$ stands for a 
variant of the modified post correspondence problem such that we will
provide the start and finish dominoes in the problem instance. This
slight change does not affect the decidability of the problem in any
way, i.e., $GPCP$ is also undecidable ~\cite{EKR82,GPCP}.

\begin{description}[align=left]
\item[Given:] A finite set of tuples $\left\{(x_i,\; y_i)\right \}^{n+1}_{i=0}$ such 
  that each $x_i, y_i \in \Sigma^+$, i.e., for   all $i$, $|x_i|>0$, $|y_i|>0$, and $(x_0, y_0), (x_{n+1}, y_{n+1})$
  are the \emph{start} and \emph{end} dominoes, respectively.
\item[Question:] Does there exist a sequence of indices $i_1, \ldots ,i_k$ such 
that \[ x_{0}\;x_{i_1}\; \ldots \;x_{i_k}\;x_{n+1} = y_{0}\;y_{i_1}\; \ldots \; y_{i_k}\;y_{n+1} ? \]
\end{description}

We work towards showing that the CT problem defined above is
undecidable by a many-one reduction from $GPCP$. First, we show how
to construct a string-rewriting system that is dwindling and
convergent from a given instance of $GPCP$.

Let $\left\{(x_i,\; y_i)\right \}^{n}_{i=1}$ the set of
``intermediate" dominoes and $(x_0, y_0), (x_{n+1}, y_{n+1})$, the
start and end dominoes respectively, be given. Suppose $\Sigma$ is the
alphabet given in the instance of $GPCP$. Without loss of
generality, we may assume $\Sigma = \{a,\, b\}$. Then set $\hat{\Sigma}
:= \{a,b\} \cup \{c_0,\;..\;c_{n+1}\} \cup \{\cent_1, \cent_2, B, a_1,
a_2, a_3, b_1, b_2, b_3\} $ which will be our alphabet for the
instance of CT.

Next we define a set of string homomorphisms used to simplify the discussion
of the reduction. Namely, we have the following:
\begin{equation*}
\begin{aligned}[c]
h_1(a) & = & a_1 \, a_2 \, a_3,\\
h_1(b) & = & b_1\, b_2 \, b_3,
\end{aligned}
\qquad
\begin{aligned}[c]
h_2(a) & = & a_1 \, a_2, \\
h_2(b) & = & b_1 \, b_2, 
\end{aligned}
\qquad
\begin{aligned}[c]
h_3(a) & = & a_1 \\
h_3(b) & = & b_1
\end{aligned}
\end{equation*}
such that each $h_i : \Sigma \rightarrow \hat{\Sigma}^+$ is a homomorphism. 

We are now in a position to construct the string rewriting system
$S$, with the following collections of rules, named as the Class~D
rules:

\begin{equation*}
\begin{aligned}[c]
\cent_1 h_1(a) & \rightarrow & \cent_1 h_3(a),\\
\cent_1 h_1(b) & \rightarrow & \cent_1 h_3(b),
\end{aligned}
\qquad
\begin{aligned}[c]
\cent_2 h_1(a) & \rightarrow & \cent_2 h_2(a)\\
\cent_2 h_1(b) & \rightarrow & \cent_2 h_2(b)
\end{aligned}
\end{equation*}

\noindent and, 
\begin{equation*}
\begin{aligned}[c]
h_i(a)\,h_1(a) & \rightarrow & h_i(a)\,h_i(a),\\
h_i(b)\,h_1(a) & \rightarrow & h_i(b)\,h_i(a),
\end{aligned}
\qquad
\begin{aligned}[c]
h_i(a)\,h_1(b) & \rightarrow & h_i(a)\,h_i(b)\\
h_i(b)\,h_1(b) & \rightarrow & h_i(b)\,h_i(b)
\end{aligned}
\end{equation*}
\noindent for $i \in \{2, 3\}$.

The erasing rules of our system consists of three classes. Class
\rom{1} rules are defined as:
\begin{eqnarray*}
\cent_1\,h_3(x_0)\,B\,c_0 & \rightarrow & \lambda \\
\cent_2\,h_2(y_0)\,c_0 & \rightarrow & \lambda
\end{eqnarray*}

and Class \rom{2} rules (for each $i = 1, 2, \ldots, n$),
\begin{eqnarray*}
	h_3(x_i)\,B\,c_i & \rightarrow & \lambda \\
	h_2(y_i)\,c_i\,B & \rightarrow & \lambda
\end{eqnarray*}

and finally Class \rom{3} rules,
\begin{eqnarray*}
	h_3(x_{n+1})\,c_{n+1} & \rightarrow & \lambda \\	
	h_2(y_{n+1})\,c_{n+1}\,B & \rightarrow & \lambda
\end{eqnarray*}

Clearly, given an instance of $GPCP$, the above set of rules can
effectively be constructed from the instance data. Also, by
inspection, we have that our system is confluent (there are no
overlaps between left-hand sides of any rules), terminating, and
dwindling.

We then set $\alpha = \cent_1$ and $\beta = \cent_2$ to complete the constructed
instance of $CT$ from $GPCP$. 

It remains to show that this instance of $CT$ is a ``yes" instance if
and only if the given instance of $GPCP$ is a ``yes" instance, i.e.,
the $CT$ has a solution if and only if the $GPCP$ does. In
that direction, we prove some results relating to~$S$.

\begin{lemma}\label{FirstStep}
Suppose $\cent_1 h_3(w_1) B \gamma \rightarrow^{!} \lambda$ and
$\cent_2 h_2(w_2) \gamma \rightarrow^{!} \lambda$
for some $w_1, w_2 \, \in \, \{a,b\}_{}^*$,
then $\gamma \in
\{c_1B,\;c_2B,\;...\;,c_nB\}_{}^{*} c_0$.
\end{lemma}

\begin{proof}
Suppose $\gamma$ is a minimal counter example with respect to length
and $\gamma \in IRR(R)$. In order for the terms to be reducible,
$\gamma = c_iB\; \gamma^\prime$ (this follows by inspection of $S$).
After we replace the $\gamma$ at the equation in the lemma, we get:
\begin{eqnarray*}
\cent_1 \;h_3(w_1)\; B\; c_i\; B\; \gamma^\prime & \rightarrow & \cent_1 {h_3(w_1)}^\prime B\; \gamma^\prime \rightarrow^{!} \lambda \\
\cent_2 \;h_2(w_2)\; c_i\; B\; \gamma^\prime & \rightarrow & \cent_2 {h_2(w_2)}^\prime \; \gamma^\prime \rightarrow^{!} \lambda
\end{eqnarray*}
\noindent
by applying the Class \rom{2} rules and finally Class \rom{1} rule to
erase the $\cent$ signs.  Then, however, $\gamma'$ is also a
counterexample, and $| \gamma' | < | \gamma |$, which is a
contradiction.
\end{proof}

We are now in a position to state and prove the main result of 
this section. 

\begin{theorem}
The CT problem is undecidable for dwindling convergent string-rewriting systems.
\end{theorem}

\begin{proof}
We first complete the ``only if" direction. Suppose CT has a solution
such that $\cent_1 Z \downarrow \cent_2 Z$ where $Z$ is a minimal
solution. We show that $Z$ corresponds to a solution for $GPCP$. Let $Z = h_1(Z_1)Z_2$ such that $h_1(Z_1)$ is the longest
prefix of $Z$ such that the following relationship holds: $Z = Z^{\prime}\; Z_2$
and $Z^{\prime} = h_1(Z_1)$ for some string $Z_1$.

$h_1(Z_1)$ can be rewritten to $h_3(Z_1)$ and $h_2(Z_1)$ by applying
the Class D rules. Thus, we will get
\begin{eqnarray*}
\cent_1 \;h_1(Z_1)\; Z_2 & \rightarrow^{*} & \cent_1 h_3(Z_1)\; Z_2 \\
\cent_2 \;h_1(Z_1)\; Z_2 & \rightarrow^{*} & \cent_2 h_2(Z_1)\; Z_2
\end{eqnarray*}
In order for terms to be reducible simultaneously, $Z_2$ must be of
the form $Z_2 = c_{n+1}\;B\;Z_2^{\prime}$. Thus
\begin{eqnarray*}
\cent_1 \;h_3(Z_1)\; Z_2 & = & \cent_1 h_3(Z_1)\; c_{n+1}\;B\;Z_2^{\prime} \\
\cent_2 \;h_2(Z_1)\; Z_2 & = & \cent_2 h_2(Z_1)\; c_{n+1}\;B\;Z_2^{\prime}
\end{eqnarray*}
i.e., $Z_1 = Z_1^\prime \: x_{n+1}$ and $Z_1 = Z_1^{\prime\prime} \: y_{n+1}$.
By applying the Class \rom{3} rules, these equations will reduce to:
\begin{eqnarray*}
\cent_1 h_3(Z_1)\; c_{n+1}\;B\;Z_2^{\prime} & \rightarrow & \cent_1 h_3(Z_1^\prime) \; B\;Z_2^{\prime}\\
\cent_2 h_2(Z_1)\; c_{n+1}\;B\;Z_2^{\prime} & \rightarrow & \cent_2 h_2(Z_1^{\prime\prime}) \; Z_2^{\prime}
\end{eqnarray*}
We now apply Lemma~\ref{FirstStep} to conclude that $Z_2^{\prime} \in
\{c_1B,\;c_2B,\; \ldots,c_nB\}_{}^{*}c_0$.\\[-5pt]

At this point we have that:
\[ Z_2 = c_{n+1}Bc_{i_1}Bc_{i_2}{\cdots}Bc_{i_k}Bc_0 \text{\:\:\: for some } i_1, \ldots, i_k\]
Then the sequence of dominoes
\[(x_0, y_0), (x_{i_k}, y_{i_k}), {\ldots} , (x_{i_1}, y_{i_1}), (x_{n+1}, y_{n+1}) \]
will be a solution to the given instance of $GPCP$ with solution
string~$Z_1$ since the left-hand sides of the Class~\rom{1}, \rom{2}, \rom{3}
rules consist of the images of domino strings under $h_2$ and
$h_3$. More specifically, there is a finite number of $B$'s and $c_i$'s
in~$Z_2$, so there must be a decomposition of $h_1(Z_1)$:
\[ h_1(Z_1) = h_1(x_0)h_1(x_{i_1}) \cdots h_1(x_{i_k})h_1(x_{n+1}) \]
and
\[ h_1(Z_1) = h_1(y_0)h_1(y_{i_1}) \cdots h_1(y_{i_k})h_1(y_{n+1}) \]
Thus, we have the following reductions with Class~D rules:
\[ \cent_1\;h_1(Z_1)\;Z_2 \rightarrow^{*} \cent_1\;h_3(Z_1)\;Z_2 \]
\[ \cent_2\;h_1(Z_1)\;Z_2 \rightarrow^{*} \cent_2\;h_2(Z_1)\;Z_2 \]
Finally, by Class  \rom{1},  \rom{2},  \rom{3} rules:
\[\cent_1\;h_3(Z_1)\;Z_2 \rightarrow^{*} \cent_1\; h_3(x_0)\;B\;c_0 \rightarrow \lambda \]
\[\cent_2\;h_2(Z_1)\;Z_2 \rightarrow^{*} \cent_2\; h_2(y_0)\;c_0 \rightarrow \lambda \]
and $Z_1$ is a solution to the instance of the~$GPCP$.

\vspace{0.1in}
We next prove the ``if" direction. Assume that the given instance of $GPCP$ has
a solution. Let $w$ be the string corresponding to the matching dominoes, and let
\[(x_0, y_0), (x_{i_1}, y_{i_1}), {\ldots} , (x_{i_k}, y_{i_k}), (x_{n+1}, y_{n+1}) \]
be the sequence of tiles that induces the match.  Let $Z =
c_{n+1}Bc_{i_1}Bc_{i_2}{\cdots}Bc_{i_k}Bc_0$. We show that
$\cent_1 h_1(w)Z\; \downarrow\; \cent_2 h_1(w)Z$.

First apply the Class $D$ rules to get: 
\[ {\cent_1}h_1(w)Z \rightarrow^{*} {\cent_1}h_3(w)Z \]
\[ {\cent_2}h_1(w)Z \rightarrow^{*} {\cent_2}h_2(w)Z \]

but then we can apply Class~\rom{1}, \rom{2}, \rom{3} rules to reduce both 
of the above terms to $\lambda$. 
\end{proof}

\vspace{0.05in}
\noindent
This result strengthens the earlier undecidability result of Otto
for string-rewriting systems that are \emph{length-reducing} and convergent.

%\pagebreak

\section{Common Equation Problem}
  
For the class of string rewriting systems the common equation problem is equivalent 
to the following problem:

\begin{description}[align=left]
\item[Input:] A string-rewriting system~$R$ on an alphabet~$\Sigma$, and 
strings~$\alpha_1, \alpha_2, \beta_1, \beta_2 \in \Sigma_{}^*$.
\item[Question:] Do there exist strings~$W_1, W_2 \in \Sigma^*$ such that 
$\alpha_1 W_1 \; \stackrel{*}{{\longleftrightarrow}_R} \; \alpha_2 W_2$ and
$\beta_1 W_1 \; \stackrel{*}{{\longleftrightarrow}_R} \; \beta_2 W_2$ 
\end{description}
  
This problem is also undecidable for the dwindling systems. The construction
we used for the CT case works here as well, since 
if $\alpha_2 = \beta_2 = \lambda$, then \[
\alpha_1 W_1 \; \stackrel{*}{{\longleftrightarrow}_R} \; W_2 \;
\stackrel{*}{{\longleftrightarrow}_R} \;
\beta_1 W_1 \] (This also shows 
that, in the string-rewriting case, CT is a particular case of~CE.)

For monadic and convergent string rewriting systems, the Common
Equation (CE) problem is decidable. This can be
shown using Lemma~3.6 in~\cite{OND98}. (See also
Theorem~3.11 of~\cite{OND98}.)
%\todo{This has to be elaborated.}

%\input{appendix}

%\nocite{DBLP:journals/jacm/BuntineB94,DBLP:journals/jsc/ComonL89}

\pagebreak

\bibliography{unifs-duals}{}

\begin{thebibliography}{10}

\bibitem{abadi2006deciding}
Mart{\'\i}n Abadi and V{\'e}ronique Cortier.
\newblock Deciding knowledge in security protocols under equational theories.
\newblock {\em Theoretical Computer Science}, 367(1-2):2--32, 2006.

\bibitem{Term}
Franz Baader and Tobias Nipkow.
\newblock {\em Term Rewriting and All That}.
\newblock Cambridge {U}niversity {P}ress, 1999.

\bibitem{Baudet2005}
Mathieu Baudet.
\newblock Deciding security of protocols against off-line guessing attacks.
\newblock In {\em Proceedings of the 12th ACM Conference on Computer and
  Communications Security}, CCS '05, pages 16--25, New York, NY, USA, 2005.
  ACM.

\bibitem{Baudet2009}
Mathieu Baudet, V{\'e}ronique Cortier, and St{\'e}phanie Delaune.
\newblock {\em YAPA: A Generic Tool for Computing Intruder Knowledge}.
\newblock Springer Berlin Heidelberg, Berlin, Heidelberg, 2009.

\bibitem{Botto}
Ronald~V Book and Friedrich Otto.
\newblock {\em String-rewriting systems}.
\newblock Springer, 1993.

\bibitem{chevalier2010compiling}
Yannick Chevalier and Micha{\"e}l Rusinowitch.
\newblock Compiling and securing cryptographic protocols.
\newblock {\em Information Processing Letters}, 110(3):116--122, 2010.

\bibitem{ciocaba2009}
{\c{S}}tefan Ciob{\^a}c{\u{a}}, St{\'e}phanie Delaune, and Steve Kremer.
\newblock {\em Computing Knowledge in Security Protocols under Convergent
  Equational Theories}.
\newblock Springer Berlin Heidelberg, Berlin, Heidelberg, 2009.

\bibitem{cortier2009}
V.~Cortier and S.~Delaune.
\newblock A method for proving observational equivalence.
\newblock In {\em 2009 22nd IEEE Computer Security Foundations Symposium},
  pages 266--276, July 2009.

\bibitem{EKR82}
Andrzej Ehrenfeucht, Juhani Karhum{\"{a}}ki, and Grzegorz Rozenberg.
\newblock The ({G}eneralized) {P}ost {C}orrespondence {P}roblem with lists
  consisting of two words is decidable.
\newblock {\em Theoretical Computer Science}, 21:119--144, 1982.

\bibitem{Meier2013}
Simon Meier, Benedikt Schmidt, Cas Cremers, and David Basin.
\newblock {\em The TAMARIN Prover for the Symbolic Analysis of Security
  Protocols}.
\newblock Springer Berlin Heidelberg, Berlin, Heidelberg, 2013.

\bibitem{NOW}
Paliath Narendran, Friedrich Otto, and Karl Winklmann.
\newblock The uniform conjugacy problem for finite {C}hurch-{R}osser {T}hue
  systems is {NP}-complete.
\newblock {\em Information and Control}, 63(1/2):58--66, 1984.

\bibitem{GPCP}
Fran{\c{c}}ois Nicolas.
\newblock ({G}eneralized) {P}ost {C}orrespondence {P}roblem and semi-{T}hue
  systems.
\newblock {\em CoRR}, abs/0802.0726, 2008.

\bibitem{OND98}
Friedrich Otto, Paliath Narendran, and Daniel~J. Dougherty.
\newblock Equational unification, word unification, and 2nd-order equational
  unification.
\newblock {\em Theoretical Computer Science}, 198(1-2):1--47, 1998.

\end{thebibliography}
\bibliographystyle{plain}

\pagebreak
\section*{Appendix}
\noindent
The following term rewriting system $R_1^{{lin}}$ specifies a fragment of linear arithmetic using
\emph{successor} and \emph{predecessor} operators:
\begin{eqnarray*}
 x - 0 & \rightarrow & x \\
 x - x & \rightarrow & 0 \\
 s(x) - y & \rightarrow & s(x-y)\\
 p(x) - y & \rightarrow & p(x-y)\\
 x - p(y) & \rightarrow & s(x-y)\\
 x - s(y) & \rightarrow & p(x-y)\\
 p(s(x)) & \rightarrow & x\\
 s(p(x)) & \rightarrow & x
\end{eqnarray*}
This TRS is convergent.

\end{document}